\documentclass[11pt]{article}

\usepackage[english]{babel}

\usepackage[margin=1in]{geometry}

\usepackage{ragged2e} 
\usepackage{tabularx}
\usepackage{booktabs}
\usepackage{array} 

\usepackage{graphicx}
\usepackage{amsfonts,amsmath,amssymb,amsthm} 
\usepackage{comment}
\usepackage{xcolor}
\definecolor{ForestGreen}{rgb}{0.1333,0.5451,0.1333}
\definecolor{DarkRed}{rgb}{0.80,0,0}
\definecolor{Red}{rgb}{1,0,0}
\usepackage[linktocpage=true,
pagebackref=true,colorlinks,
linkcolor=DarkRed,citecolor=ForestGreen,
bookmarks,bookmarksopen,bookmarksnumbered]
{hyperref}
\usepackage{cleveref}
\usepackage{algorithm}
\usepackage[noend]{algpseudocode}
\usepackage{booktabs}
\usepackage{multirow}
\usepackage{color-edits}
\addauthor[Yingxi]{yl}{blue}
\addauthor[Mingwei]{mw}{red}
\addauthor[Ellen]{ev}{purple}

\theoremstyle{plain}
\newtheorem{theorem}{Theorem}[section]
\newtheorem{lemma}[theorem]{Lemma}

\newtheorem{corollary}[theorem]{Corollary}
\newtheorem{claim}[theorem]{Claim}
\newtheorem{fact}[theorem]{Fact}

\theoremstyle{definition}
\newtheorem{definition}[theorem]{Definition}

\newcommand{\R}{\mathbb{R}}

\newcommand{\E}{\mathbb{E}}
\newcommand{\D}{\mathbb{D}}
\newcommand{\U}{\mathbb{U}}

\newcommand{\calH}{\mathcal{H}}

\newcommand{\calF}{\mathcal{F}}
\newcommand{\calT}{\mathcal{T}}

\newcommand{\calX}{\mathcal{X}}

\newcommand{\calA}{\mathcal{A}}

\newcommand{\bfp}{\mathbf{p}}
\newcommand{\bfq}{\mathbf{q}}

\newcommand{\bfw}{\mathbf{w}}

\newcommand{\norm}[1]{\left\|{#1} \right\|}
\newcommand{\indc}[1]{{\mathbf{1}\left\{{#1}\right\}}}

\newcommand{\cost}{\mathrm{cost}}
\newcommand{\opt}{\mathrm{OPT}}
\newcommand{\ber}{\mathrm{Ber}}
\newcommand{\bin}{\mathrm{Bin}}
\newcommand{\std}{\mathrm{std}}

\newcommand{\pb}{\mathrm{PB}}

\newcommand{\RS}{\mathrm{RS}}
\newcommand{\BBGN}{\mathrm{BBGN}}
\newcommand{\diam}{\mathrm{diam}}

\newcommand{\floor}[1]{\left \lfloor #1 \right \rfloor}

\title{Smoothed Analysis of Online Metric Matching with a Single Sample: Beyond Metric Distortion}

\author{
    Yingxi Li \\
    Stanford University \\
    \texttt{yingxi@stanford.edu}
    \and
    Ellen Vitercik \\
    Stanford University \\
    \texttt{vitercik@stanford.edu}
    \and
    Mingwei Yang \\
    Stanford University \\
    \texttt{mwyang@stanford.edu}
}
\date{\today}

\begin{document}
\maketitle

\begin{abstract}
    In the online metric matching problem, $n$ servers and $n$ requests lie in a metric space.
    Servers are available upfront, and requests arrive sequentially.
    An arriving request must be matched immediately and irrevocably to an available server, incurring a cost equal to their distance. 
    The goal is to minimize the total matching cost.

    We study this problem in $[0, 1]^d$ with the Euclidean metric, when servers are adversarial and requests are independently drawn from distinct distributions that satisfy a mild smoothness condition. Our main result is an $O(1)$-competitive algorithm for $d \neq 2$ that requires no distributional knowledge, relying only on a single sample from each request distribution.
    To our knowledge, this is the first algorithm to achieve an $o(\log n)$ competitive ratio for non-trivial metrics beyond the i.i.d. setting.
    Our approach bypasses the $\Omega(\log n)$ barrier introduced by probabilistic metric embeddings: instead of analyzing the embedding distortion and the algorithm separately, we directly bound the cost of the algorithm on the target metric space of a simple deterministic embedding.
    We then combine this analysis with lower bounds on the offline optimum for Euclidean metrics, derived via majorization arguments, to obtain our guarantees.
\end{abstract}
\section{Introduction}

The online metric matching problem is a classic topic in the design of online algorithms and has been studied extensively for decades.
In this problem, $n$ servers are available upfront and $n$ requests arrive sequentially, with all servers and requests lying in a common metric space.
Each request must be immediately and irrevocably matched to an unmatched server, incurring a cost equal to their distance.
The goal is to minimize the total matching cost.

The online metric matching problem captures a variety of practical scenarios. In ride-hailing, for example, servers and requests correspond to drivers and passengers, and the cost of a match reflects the pickup distance. The Euclidean special case is also natural for applications such as kidney exchange, where patients and donors are represented by high-dimensional feature vectors and the Euclidean distance between them captures compatibility.

When servers and requests are adversarially chosen, \cite{DBLP:conf/soda/MeyersonNP06} give an $O(\log^3 n)$-competitive algorithm for general metrics, later improved to $O(\log^2 n)$~\cite{DBLP:journals/algorithmica/BansalBGN14}.
For Euclidean metrics, an $O(\log n)$-competitive algorithm is given by~\cite{DBLP:conf/icalp/GuptaL12}.
On the other hand, there exists an $\Omega(\log n)$ competitive-ratio lower bound for uniform metrics\footnote{A metric is \emph{uniform} if every pair of distinct points has distance $1$.}~\cite{DBLP:conf/soda/MeyersonNP06}, which holds even in the random-order arrival model~\cite{DBLP:conf/approx/Raghvendra16}. Moreover, no algorithm can achieve an $o(\sqrt{\log n})$ competitive ratio for line metrics~\cite{DBLP:journals/talg/PesericoS23}.

The fully adversarial model is often considered overly pessimistic, since worst-case scenarios rarely arise in practice.
A more realistic assumption is that requests are i.i.d. sampled, while servers remain adversarially chosen.
In this setting, \cite{DBLP:conf/icalp/GuptaGPW19} give an $O((\log \log \log n)^2)$-competitive algorithm for general metrics and distributions, and their algorithm is $O(1)$-competitive for tree metrics.
Recently, \cite{DBLP:journals/ior/YangY26} show that, at the cost of a constant factor in the competitive ratio, the above setting further reduces to the case where all servers and requests are i.i.d. sampled from the same distribution. Building on this, they obtain an $O(1)$-competitive algorithm for Euclidean metrics with \emph{smooth} distributions, where a distribution is smooth if it admits a density with respect to the uniform measure upper bounded by a constant.

The i.i.d. model yields nearly tight guarantees, but it is often unrealistic: in many applications, different requests follow markedly different distributions.
A more promising approach is \emph{smoothed analysis}, which has been widely successful in both theory and practice, and is frequently used to explain the strong empirical performance of heuristics with poor worst-case guarantees~\cite{DBLP:journals/cacm/SpielmanT09,DBLP:conf/soda/ChenGVY24,DBLP:conf/stoc/Anastos0M25}.
In this model, the adversary first selects an adversarial input, which is then randomly perturbed by nature via adding a small noise, and the algorithm is required to perform well in expectation.
Since the perturbed input always follows a smooth distribution, in the more modern and general formulation of smoothed analysis, the adversary directly specifies a smooth distribution over inputs~\cite{DBLP:journals/jacm/HaghtalabRS24,DBLP:conf/esa/CoesterU25}.
In this paper, we adopt the smoothed analysis framework to study online metric matching and obtain improved competitive ratios under significantly weaker assumptions than those in the i.i.d. setting.

\subsection{Our Results}

\paragraph{Main result.} We focus on $[0, 1]^d$ with the Euclidean metric, assuming that servers are adversarial and requests are independently drawn from distinct smooth distributions.
As our main result, we give an algorithm which requires no knowledge of the distributions and is $O(1)$-competitive for $d \neq 2$, given access to one sample from each of the $n$ request distributions (\Cref{thm:cp-euc}).
In particular, the algorithm does not need to know the correspondence between distributions and samples, and its competitive ratio depends polynomially on $d$ and the smoothness parameter, which are usually treated as constants that do not grow with $n$~\cite{DBLP:conf/icalp/GuptaL12,DBLP:conf/sigecom/Kanoria22,DBLP:journals/ior/YangY26}.
To the best of our knowledge, this is the first algorithm to achieve an $o(\log n)$ competitive ratio for non-trivial metrics in a setting beyond the i.i.d. assumption, even when the request distributions are fully known to the algorithm (see \Cref{tab:comparison} for a thorough comparison with prior work).

\begin{table}[t]
\centering
\small
\begin{tabularx}{\textwidth}{
  @{}>{\raggedright\arraybackslash}m{2.2cm}
  >{\raggedright\arraybackslash}m{2cm}
  >{\raggedright\arraybackslash}m{2.2cm}
  >{\raggedright\arraybackslash}m{2.8cm}
  >{\raggedright\arraybackslash}m{2.8cm}
  >{\raggedright\arraybackslash}m{2cm}@{}}
\toprule
\textbf{Requests} & \textbf{Servers} & \textbf{Metric Space} & \textbf{Distributional Knowledge} & \textbf{Competitive Ratio} & \textbf{Reference} \\
\midrule
Adversarial
& Adversarial
& General
& None
& $O(\log^2 n)$
& \cite{DBLP:journals/algorithmica/BansalBGN14} \\
\addlinespace[2pt]
Adversarial
& Adversarial
& Euclidean
& None
& $O(\log n)$
& \cite{DBLP:conf/icalp/GuptaL12} \\
\addlinespace[2pt]
i.i.d.
& Adversarial
& General
& Full knowledge
& $O((\log \log \log n)^2)$ 
& \cite{DBLP:conf/icalp/GuptaGPW19} \\
\addlinespace[2pt]
i.i.d. 
& Adversarial
& Tree
& Full knowledge
& $O(1)$
& \cite{DBLP:conf/icalp/GuptaGPW19} \\
\addlinespace[2pt]
i.i.d. smooth
& Adversarial
& Euclidean
& $O(n^2)$ samples
& $O(1)$ \textbf{for} $d \neq 2$
& \cite{DBLP:journals/ior/YangY26} \\
\addlinespace[2pt]
i.i.d. uniform
& Adversarial
& Euclidean
& Full knowledge
& $O(1)$
& \cite{DBLP:journals/ior/YangY26} \\
\midrule
Independent, non-identical, smooth
& Adversarial
& Euclidean
& One sample per distribution
& \textbf{$O(1)$ for $d\neq 2$}
& \textbf{Our work} \\
\bottomrule
\end{tabularx}
\caption{Comparison with prior work on online metric matching {with the state-of-the-art results}.  Our result is the first to achieve an $o(\log n)$ competitive ratio for non-trivial metrics under distributional assumptions strictly weaker than i.i.d.}
\label{tab:comparison}
\end{table}

\paragraph{A complementary algorithm.} In addition, we present a second algorithm under the same assumptions whose guarantees are not directly comparable (\Cref{thm:cp-euc}).
It is $O(1)$-competitive for $d \ge 3$ and, while still does not improve upon the $O(\log n)$ competitive ratio for $d = 2$ achieved in the adversarial setting, it outperforms the first algorithm in the regime of $d = 2$.

\paragraph{Discussion.} 
Our sample-based assumption---one sample per request distribution---is a natural and practical assumption, widely adopted in recent work on online algorithms~\cite{DBLP:conf/soda/AzarKW14,DBLP:conf/innovations/RubinsteinWW20,DBLP:conf/soda/KaplanNR22,DBLP:conf/focs/DuttingKLR024,DBLP:conf/stoc/Ghuge0W25}.
This assumption is also necessary to some extent, as it is unclear how an algorithm could leverage the distributional properties of requests without some information, even in simpler i.i.d. settings.
In comparison, our assumptions are weaker than those in prior work for the i.i.d. model: the algorithm of \cite{DBLP:conf/icalp/GuptaGPW19} requires full knowledge of the distribution, and the algorithm of \cite{DBLP:journals/ior/YangY26} needs $O(n^2)$ samples.

Finally, we turn to the role of the dimension $d$.
Two-dimensional space is a critical case for Euclidean matching: prior works highlight that the plane behaves fundamentally differently from both the line and higher dimensions~\cite{talagrand2022upper,DBLP:conf/sigecom/Kanoria22,DBLP:journals/ior/YangY26}. Consistent with this, our analysis does not yield a constant-competitive bound in $d = 2$. Nevertheless, our second algorithm achieves strictly better performance in two dimensions than our main algorithm, offering partial progress. Closing the gap in two dimensions remains an important direction for future research.

\subsection{Technical Overview}
Our approach builds on the classical paradigm of metric embeddings but departs in a crucial way.
Previous algorithms embed the input metric space into \emph{Hierarchically well-Separated Trees (HST)} and analyze two pieces separately: the distortion of the embedding and the competitive ratio of the algorithm on HSTs~\cite{DBLP:conf/soda/MeyersonNP06,DBLP:conf/icalp/GuptaL12,DBLP:journals/algorithmica/BansalBGN14}.
This separation is what forces the $\Omega(\log n)$ barrier. Our key idea is to bypass the distortion step entirely: we analyze the algorithm's cost directly in the resulting HST, using the non-contractivity of the embedding to argue about the original metric space.
This avoids the logarithmic loss while retaining the algorithmic utility of the HST. With this perspective in place, we now recall the definitions of metric embeddings and HSTs.

A \emph{(non-contractive) deterministic embedding} of a metric space $(\calX, \delta)$ into another metric space $(\calX', \delta')$ is a mapping $f: \calX \to \calX'$ satisfying $\delta(x, y) \leq \delta'(f(x), f(y))$ for all $x, y \in \calX$, and the \emph{distortion} of an embedding $f$ is the smallest $\kappa \geq 1$ such that $\delta'(f(x), f(y)) \leq \kappa \cdot \delta(x, y)$ for all $x, y \in \calX$.
A \emph{probabilistic embedding} is a distribution over deterministic embeddings, and the distortion of a probabilistic embedding $f$ is the smallest $\kappa \geq 1$ such that $\E[\delta'(f(x), f(y))] \leq \kappa \cdot \delta(x, y)$ for all $x, y \in \calX$.
It is not hard to see that given an embedding of the input metric space to another (usually simpler) metric space with distortion $\kappa_1$, and an algorithm for the latter metric space with competitive ratio $\kappa_2$, we get an algorithm for the input metric space with competitive ratio $\kappa_1 \cdot \kappa_2$.
In other words, a metric embedding with distortion $\kappa$ reduces the problem for a complicated metric to the same problem for a simple metric with the cost of an additional factor $\kappa$ in the competitive ratio.

One of the most popular simple metrics is the HST metric, which is induced by the distance function on an HST.
In particular, an $\alpha$-HST for $\alpha \geq 1$ is a rooted tree where all the leaf nodes are at the same depth. Its edge lengths are defined hierarchically: assuming root-adjacent edges have length $1$, and for every internal node $v$, the edge to its parent is $\alpha$ times longer than each edge to its children (see \Cref{def:hst} for a formal definition and \Cref{fig:hst} for an illustration).
It is known that every $n$-point metric space can be randomly embedded into an $\alpha$-HST with distortion $O(\alpha \log n)$~\cite{DBLP:journals/jcss/FakcharoenpholRT04}, which, when used as a reduction for online problems, contributes an $O(\log n)$ factor to the competitive ratio.

To bypass the $O(\log n)$ distortion factor, we avoid the usual two-step analysis that handles the distortion of the metric embedding and the competitive ratio for HSTs separately.
Instead, we bound directly the cost the algorithm incurred on the resulting HST; by the non-contractivity of the embedding, this also upper-bounds the algorithm’s cost in the original metric space.
We then compare this upper bound to a lower bound for the offline optimum in the original metric space, yielding the desired competitive ratio. Notably, a simple deterministic embedding---despite having unbounded worst-case distortion---suffices for our analysis.

In more detail, we use the canonical dyadic partition of $[0, 1]^d$ to define a $2^d$-ary $2$-HST of height $h$ (a tunable parameter we optimize later), where an HST is $\Delta$-ary if every internal node has exactly $\Delta$ children.
To describe the embedding, the root of the HST corresponds to $[0, 1]^d$, each node is partitioned into $2^d$ subcubes with half the side-length, and the leaves at depth $h$ correspond to disjoint subcubes of side-length $2^{-h}$ (see \Cref{fig:hg} for an illustration).
Then, we map each point $x \in [0, 1]^d$ to the unique leaf node whose corresponding cube contains $x$.

Equipped with the metric embedding, the problem reduces to designing algorithms for HST metrics, where we adopt the \emph{Random-Subtree} algorithm of \cite{DBLP:conf/icalp/GuptaL12} and the algorithm of \cite{DBLP:journals/algorithmica/BansalBGN14}.
Our key observation is that the expected cost analysis for both algorithms reduces to bounding the fluctuations in the number of requests in each subtree.
By controlling these fluctuations (via standard deviations of Poisson–binomial counts) and applying concavity, we show that the worst case occurs when requests are uniformly distributed across children. This yields distribution-free upper bounds for the cost (Theorems~\ref{thm:cost-random-subtree} and~\ref{thm:cost-bbgn}), so no smoothness assumption is needed for the upper bounds.

Smoothness enters only in the lower bound on the offline optimum (\Cref{lem:lb-smooth}).
For $d \geq 2$, the bound follows directly from standard nearest-neighbor distance estimates.
The proof for $d = 1$ is subtler.
Any imbalance between the numbers of servers and requests in a subinterval of $[0, 1]$ forces a proportional number of matches to cross the subinterval's boundary, contributing to the total cost.
Inspecting all subintervals of length $L \in (0, 1)$ yields a lower bound for the offline optimum, which we refer to as \emph{obstacle to matching at length scale $L$}.
It is known that, when all servers and requests are uniformly distributed, the largest obstacle to matching for $d = 1$ occurs at the scale of a constant length~\cite{DBLP:conf/sigecom/Kanoria22}.
To generalize this reasoning to the case where requests are drawn from distinct smooth distributions, we derive a lower bound for the obstacle to matching at a length scale proportional to the smoothness parameter. Using majorization and concavity arguments, we show that this lower bound is minimized when requests are as concentrated as possible.
The desired lower bound then follows from the anti-concentration properties of smooth distributions.

Finally, we generalize the result in \cite{DBLP:journals/ior/YangY26} to incorporate the provided samples from the request distributions in a black-box manner (\Cref{lem:stochastic-to-semi}).
Combining the above upper bounds for the algorithms with the lower bounds for the offline optimum yields our main result.

\subsection{Related Work}

\paragraph{Further related work on online metric matching.}
For adversarial servers and requests, \cite{DBLP:journals/jal/KalyanasundaramP93,DBLP:journals/tcs/KhullerMV94,DBLP:conf/approx/Raghvendra16} give $(2n-1)$-competitive deterministic algorithms for general metrics, which is optimal.
\cite{DBLP:conf/approx/Raghvendra16} provides a primal-dual deterministic algorithm that is $O(\log n)$-competitive for general metrics in the random-order arrival model, which is later shown to exhibit near instance-optimality~\cite{DBLP:conf/focs/NayyarR17} and an $O(\log n)$ competitive ratio for line metrics~\cite{DBLP:conf/compgeom/Raghvendra18}.

When all servers and requests are uniformly distributed on the Euclidean space, \cite{DBLP:conf/sigecom/Kanoria22} gives an $O(1)$-competitive algorithm, which applies the same deterministic metric embedding as ours.
Nevertheless, they employ a different algorithm for the HST metric, and it is unclear how to generalize their analysis to non-identical request distributions.
When all servers and requests are i.i.d. drawn from a general distribution, \cite{DBLP:conf/sigecom/ChenKKZ23} present an algorithm for Euclidean metrics with nearly optimal \emph{regret}, which is defined as the difference between the cost of the algorithm and the offline optimum.
\cite{DBLP:conf/sigecom/BalkanskiFP23} show that the \emph{Greedy} algorithm is $O(1)$-competitive for line metrics when all servers and requests are uniformly distributed.

\cite{DBLP:conf/sigecom/AkbarpourALS22} initiate the study of \emph{unbalanced markets}, where servers outnumber requests by a constant factor, and they show that Greedy is $O(\log^3 n)$-competitive when all servers and requests are uniformly distributed on a line, which is subsequently improved to $O(1)$~\cite{DBLP:conf/sigecom/BalkanskiFP23}.
\cite{DBLP:conf/sigecom/Kanoria22} gives an $O(1)$-competitive algorithm for unbalanced markets when all servers and requests are uniformly distributed on the Euclidean space.
When servers are adversarial and requests are i.i.d. drawn from a smooth distribution, \cite{DBLP:journals/ior/YangY26} achieve competitive-ratio and regret guarantees for unbalanced markets and Euclidean metrics.

Several variants of online metric matching have also been studied, which include online transportation~\cite{DBLP:conf/icalp/HaradaI25,arndt2025competitive}, online metric matching with recourse~\cite{DBLP:conf/approx/0004KS20,DBLP:journals/algorithmica/MegowN25}, online min-cost perfect matching with delays~\cite{DBLP:conf/stoc/EmekKW16,DBLP:conf/approx/AshlagiACCGKMWW17}, online metric matching with stochastic arrivals and departures~\cite{DBLP:journals/ior/AouadS22,DBLP:conf/stoc/AmaniHamedaniAS25}, online min-weight perfect matching~\cite{DBLP:conf/soda/BhoreFT24}, and online matching in geometric random graphs~\cite{DBLP:journals/corr/abs-2306-07891}.

\paragraph{Smoothed analysis of online problems.}
For smoothed analysis of other online problems, prior studies primarily focus on online learning~\cite{balcan2018dispersion, DBLP:conf/nips/KannanMRWW18,balcan2020semi,DBLP:conf/nips/HaghtalabHSY22,DBLP:conf/colt/BlockDGR22,DBLP:conf/sigecom/DurvasulaHZ23,DBLP:journals/jacm/HaghtalabRS24} and online discrepancy minimization~\cite{DBLP:conf/innovations/BansalJM0S22,DBLP:conf/icalp/0001JM0S22,DBLP:journals/jacm/HaghtalabRS24}, whose goal is to minimize regret.
Regarding smoothed competitive analysis, \cite{DBLP:journals/mor/BecchettiLMSV06} consider the multi-level feedback algorithm for non-clairvoyant scheduling, and \cite{DBLP:journals/tcs/SchaferS05} analyze the work function algorithm for metrical task systems.
More recently, \cite{DBLP:conf/esa/CoesterU25} conduct smoothed analysis of classic online metric problems including $k$-server, $k$-taxi, and chasing small sets, where requests are drawn from smooth distributions, and they achieve significantly improved competitive ratios compared to the adversarial setting.
In particular, they allow the distribution followed by each request to depend on the realizations of the past requests and the decisions made by the algorithm thus far, and their algorithms require no distributional knowledge.
However, their technique does not results in improved algorithms for online metric matching.

\paragraph{Sample complexity of online problems.}
A growing line of work studies online algorithms that receive samples from the underlying distributions to go beyond the strong assumption of knowing the distributions in full.
Pioneered by \cite{DBLP:conf/soda/AzarKW14}, competitive guarantees are achieved for sample-based prophet inequalities under various combinations of the arrival model, the number of samples, and combinatorial constraints~\cite{DBLP:conf/ec/CorreaDFS19,DBLP:conf/innovations/RubinsteinWW20,DBLP:conf/soda/CaramanisDFFLLP22,DBLP:journals/mor/0001C0S24,DBLP:conf/stoc/CristiZ24,DBLP:conf/sigecom/0001LTWW024}.
This paradigm has also been adopted by literature on online resource allocation~\cite{DBLP:conf/focs/DuttingKLR024,DBLP:conf/stoc/Ghuge0W25} and online weighted matching~\cite{DBLP:conf/soda/KaplanNR22}.
\section{Preliminaries}

Let $(\calX, \delta)$ be a metric space.
There are $n$ servers $S = \{s_1, \ldots, s_n\}$ available at time $t = 0$, whose locations are known to the algorithm, and $n$ requests $R = (r_1, \ldots, r_n)$ that arrive sequentially.
At each time step $t \in [n]$, the location of the request $r_t$ is revealed, and the algorithm must immediately and irrevocably match it to an unmatched server.
The cost of matching a request $r$ and a server $s$ is $\delta(r, s)$.
We assume that servers are adversarial, and requests are independently drawn from distributions $\D_1, \ldots, \D_n$ support on $\calX$.
Let $\D := \prod_{i=1}^n \D_i$ be the joint request distribution, which is not known to the algorithm.

Given a (randomized) algorithm $\calA$, let $\cost(\calA; S, R)$ be the expected cost of $\calA$ for server set $S$ and request sequence $R$.
For a distribution $\D$ over request sequences, define $\cost(\calA; S, \D) := \E_{R \sim \D}[\cost(\calA; S, R)]$.
Given $S$ and $R$, let $\opt(S, R)$ be the minimum cost of a perfect matching between $S$ and $R$.
Similarly, let $\opt(S, \D) := \E_{R \sim \D} [\opt(S, R)]$ be the expected optimal cost when requests are drawn from $\D$, and $\opt(\D, \D) := \E_{S \sim \D}[\opt(S, \D)]$ be the expected optimal cost when servers and requests are all drawn from $\D$.
We say that an algorithm $\calA$ is \emph{$\alpha$-competitive} for $\alpha \geq 1$ if for all $S$ and $\D$, $\cost(\calA; S, \D) \leq \alpha \cdot \opt(S, \D)$.
Finally, given a matching $M$ between two sets $S$ and $R$, we denote the element in $R$ that is matched to $s \in S$ as $M(s)$.

\begin{definition}[Smoothness]
    We say that a measure $\mu$ over $\calX$, which supports a uniform distribution $\U$, is \emph{$\sigma$-smooth} for $\sigma \in (0, 1]$ if for every measurable subset $\calX' \subseteq \calX$, $\mu(\calX') \leq \U(\calX')/\sigma$.
\end{definition}

\subsection{Majorization and Poisson Binomial}

For a vector $\bfp \in \R^n$, we use $p_{[i]}$ for $i \in [n]$ to denote the $i$-th largest element among $p_1, \ldots, p_n$.

\begin{definition}[Majorization]
    For $\bfp, \bfq \in \R^n$, we say that $\bfp$ \emph{majorizes} $\bfq$, denoted as $\bfp \succ \bfq$, if
    \begin{itemize}
        \item $\sum_{i=1}^k p_{[i]} \geq \sum_{i=1}^k q_{[i]}$ for $k \in [n - 1]$, and
    
        \item $\sum_{i=1}^n p_{[i]} = \sum_{i=1}^n q_{[i]}$.
    \end{itemize}
\end{definition}

For $\bfp \in [0, 1]^n$, let $\pb (\bfp)$ denote the \emph{Poisson binomial} random variable $\sum_{i=1}^n X_i$ with independent $X_i \sim \ber (p_i)$.
For $\bfp, \bfq \in [0, 1]^n$ with $\bfp \succ \bfq$, it is known that $\pb(\bfp)$ is ``less spread out'' than $\pb(\bfq)$ in the sense of \emph{convex order}, which we formally present in the following lemma.

\begin{lemma}[Proposition 12.F.1 in \cite{marshall11}]\label{lemma:convex_order}
    For $\bfp, \bfq \in [0, 1]^n$, if $\bfp \succ \bfq$, then $\E[f(\pb(\bfp))] \leq \E[f(\pb(\bfq))]$ for every convex function $f: \{0, 1, \ldots, n\} \to \R$.
\end{lemma}

The following corollary, which compares the mean absolute deviations of two Poisson binomial random variables, is a direct consequence of \Cref{lemma:convex_order} since $f(t) = |t - c|$ is convex.

\begin{corollary}\label{cor:convex-order-mean-abs-devia}
    For $\bfp, \bfq \in [0, 1]^n$ with $\bfp \succ \bfq$, $\E[|\pb(\bfp) - \sum_{i=1}^n p_i|] \leq \E[|\pb(\bfq) - \sum_{i=1}^n q_i|]$.
\end{corollary}

\subsection{Algorithm with Sample Access to Request Distributions}\label{sec:adv_to_smooth}

When $\D_1, \ldots, \D_n$ are identical, \cite[Theorem 2]{DBLP:journals/ior/YangY26} shows that, given an arbitrary algorithm $\calA$ and $n$ samples from $\D_1$, we can construct another algorithm whose cost can be decomposed into the offline optimum and the cost of $\calA$ when the server set is also drawn from $\D$.
Following similar arguments, we generalize this result to the case with non-identical request distributions, where one sample from each request distribution is provided.
The proof of \Cref{lem:stochastic-to-semi} is deferred to \Cref{sec:proof-lem-stochastic-to-semi}.

\begin{lemma}\label{lem:stochastic-to-semi}
    Given an algorithm $\calA$ and one sample from each request distribution, there exists an algorithm $\calA'$ such that $\cost(\calA'; S, \D) \leq \opt(S, \D) + \cost(\calA; \D, \D)$.
    In particular, $\calA'$ does not need to know the correspondence between distributions and samples.
\end{lemma}

\subsection{Hierarchically Well-Separated Trees}\label{sec:hst}

In this subsection, we introduce Hierarchically well-Separated Trees.

\begin{definition}[HSTs]\label{def:hst}
    Given $\alpha \geq 1$, an \emph{$\alpha$-Hierarchically well-Separated Tree ($\alpha$-HST)} is a rooted tree $T = (V, E)$ along with a distance function $\delta: E \to \R_{\geq 0}$ on the edges that satisfies the following properties:
    \begin{enumerate}
        \item For each internal node $v$, all edges from $v$ to its children have the same length.
        \item For each node $v$, if $p(v)$ is the parent of $v$, and $c(v)$ is a child of $v$, then $\delta(p(v), v) = \alpha \cdot \delta(v, c(v))$.
        \item For all leaf nodes $v_1$ and $v_2$, let $p(v_1)$ and $p(v_2)$ be the parents of $v_1$ and $v_2$, respectively.
        Then, $\delta(p(v_1), v_1) = \delta(p(v_2), v_2)$.
    \end{enumerate}
    Moreover, an HST is \emph{$\Delta$-ary} if each internal node has exactly $\Delta$ children, and the \emph{height} of an HST is defined as the number of edges in the path from the root to any leaf node.
\end{definition}

\begin{figure}
    \centering
    \includegraphics[width=0.6\linewidth]{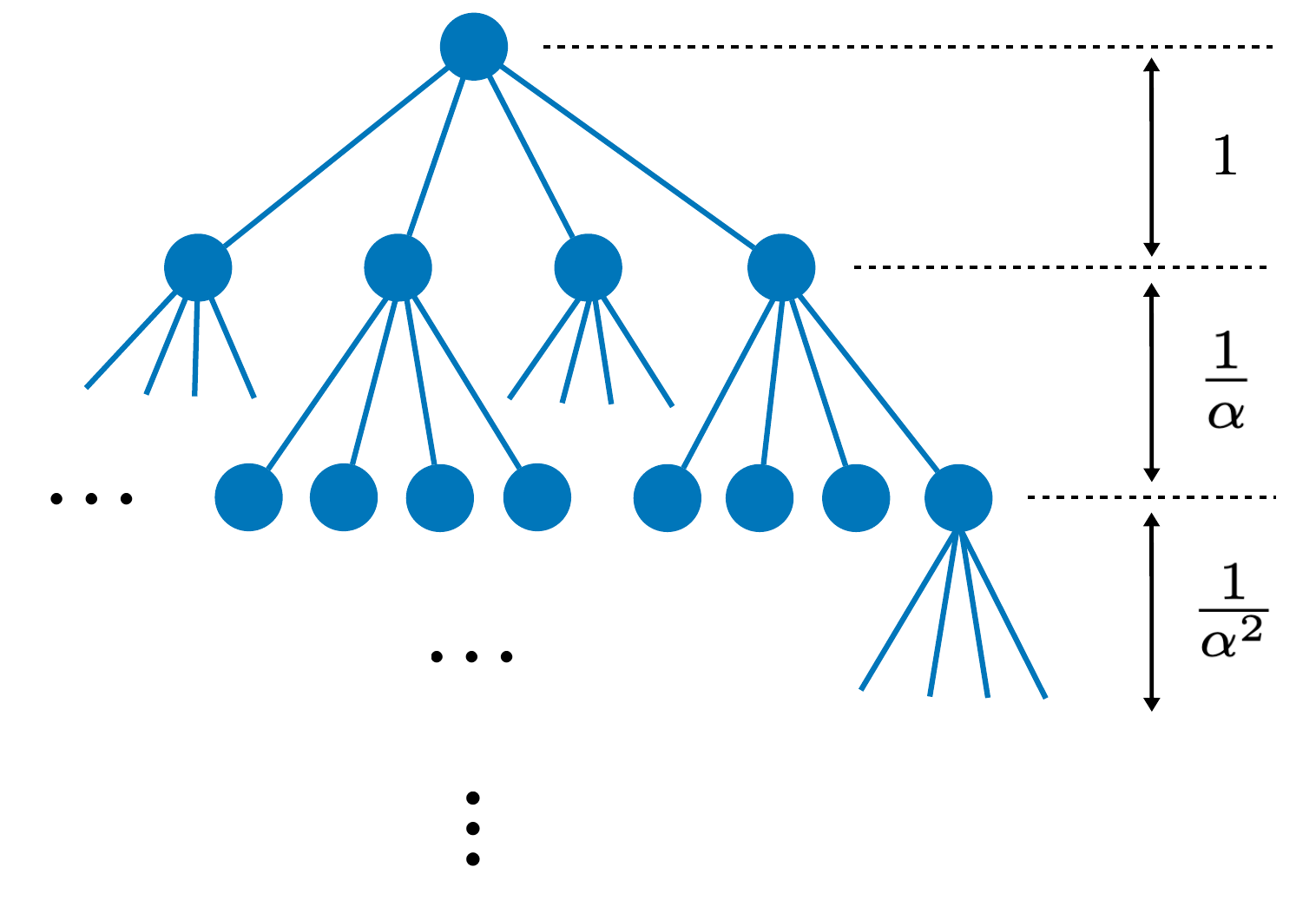}
    \caption{A 4-ary $\alpha$-HST.}
    \label{fig:hst}
\end{figure}

By normalization, we assume the length of each root-edge of an HST to be $1$. See Figure~\ref{fig:hst} for an illustration of a 4-ary $\alpha$-HST.
An HST naturally induces a metric over $V$, where the distance between two nodes is the length of the unique tree path between them.
For HST metrics, unless stated otherwise, we assume that servers and requests {are located at} the leaf nodes.

For each internal node $v$ of a $\Delta$-ary HST, let $c_i(v)$ be the $i$-th child of $v$ for $i \in [\Delta]$. Define $\hat{s}(v)$ and $\hat{r}(v)$ as the number of servers and requests, respectively, in the subtree rooted at $v$.
The height of a node $v$ is defined as the number of edges in the path between $v$ and any leaf node in the subtree rooted at $v$.
For an HST with height $h$, we use $V_j$ to denote the set of nodes with height $j$ for each $j \in \{0, 1, \ldots, h\}$.

When the request sequence is drawn from $\D = \prod_{i=1}^n \D_i$, for each node $v$, let $\mu_{\D_i}(v)$ be the probability that $r_i$ is in the subtree rooted at $v$, and let $\mu_{\D}(v) := \sum_{i=1}^n \mu_{\D_i}(v)$ be the expected number of requests in the subtree rooted at $v$.
Note that $\hat{r}(v) \sim \pb(\mu_{\D_1}(v), \ldots, \mu_{\D_n}(v))$.
\section{Random-Subtree Algorithm}

The \emph{Random-Subtree (RS)} algorithm proposed by \cite{DBLP:conf/icalp/GuptaL12} is a natural starting point for online matching on HST metrics. In the adversarial setting, it achieves an $O(\log \Delta)$ competitive ratio for $\Delta$-ary $\alpha$-HST metrics, provided $\alpha = \Omega(\log \Delta)$. Here we revisit the RS algorithm in the stochastic setting where servers and requests are independently drawn from a distribution $\D = \prod_{i=1}^n \D_i$.
Later on, we will translate the guarantees to the setting where servers are adversarial and requests are drawn from $\D$ by applying \Cref{lem:stochastic-to-semi}.
Rather than its competitive ratio, we upper bound the cost of the $\RS$ algorithm for arbitrary $\alpha \geq 2$.

We start by describing the $\RS$ algorithm, which applies to any HST metric.
For each arriving request $r$, let $v$ be the lowest ancestor of $r$ such that the subtree rooted at $v$ contains at least one available server.
Starting from $v$, the algorithm descends the tree toward a leaf guaranteed to contain an available server. At each internal node $v$, it selects uniformly at random among the children whose subtrees contain an available server, and continues this process until reaching such a leaf. The request $r$ is then matched to a server at this leaf.

We upper bound the cost of the $\RS$ algorithm in the following theorem, where we use $H_k := 1 + 1/2 + \ldots + 1 / k$ to denote the $k$-th harmonic number.

\begin{theorem}\label{thm:cost-random-subtree}
    For any $\Delta$-ary $\alpha$-HST with height $h$ and $\alpha \geq 2$, and for any $\D = \prod_{i=1}^n \D_i$,
    \begin{align*}
        \cost(\RS; \D, \D)
        \leq 6H_{\Delta} \sqrt{n} \sum_{j=0}^{h-1} \left(\frac{\sqrt{\Delta^{h-j}}}{\alpha^{h-j-1}} \sum_{\ell = 0}^{j} \xi^{\ell}\right),
    \end{align*}
    where $\xi := H_{\Delta} / \alpha$.
\end{theorem}

The bound in \Cref{thm:cost-random-subtree} admits a natural interpretation.
Suppose, for intuition, that every $\D_i$ is uniform over the leaf nodes---this turns out to be the worst case of our analysis.
To facilitate interpretation, we rewrite the bound as
\begin{align}\label{eqn:explain-rs}
    6 \sum_{j=0}^{h - 1} \sqrt{n \Delta^{h - j}} \sum_{\ell = 0}^j \frac{H_{\Delta}^{\ell + 1}}{\alpha^{h - 1 - j + \ell}}.
\end{align}
The term $\sqrt{n \Delta^{h - j}} \approx \sum_{v \in V_j} \E[|\hat{r}(v) - \hat{s}(v)|]$ corresponds to the expected discrepancy between servers and requests at level $j$, i.e., in the subtrees rooted at nodes in $V_j$.
Each discrepancy at level $j$ contributes to at most $H_{\Delta}$ mismatches made by the algorithm, and each such mismatch incurs a cost on the order of $\alpha^{j-h+1}$.
Moreover, a mismatch at level $j$ can cascade downward, creating up to $H_{\Delta}$ additional mismatches at level $j - 1$, which then trigger up to $H_{\Delta}^2$ additional mismatches at level $j-2$, and so on.
This propagation explains the second summation in \eqref{eqn:explain-rs}.

The proof of \Cref{thm:cost-random-subtree} relies on the following bound for the $\RS$ algorithm, whose cost is upper bounded in terms of the number of excess requests in each subtree.

\begin{lemma}\label{lem:framework-hst-det}
    For any $\Delta$-ary $\alpha$-HST with height $h$ and $\alpha \geq 2$, and for all $S$ and $R$,
    \begin{align*}
        \cost(\RS; S, R)
        \leq 3H_{\Delta} \sum_{j=0}^{h-1} \sum_{v \in V_j} \frac{(\hat{r}(v) - \hat{s}(v))^+}{\alpha^{h- j - 1}} \sum_{\ell = 0}^{j} \xi^{\ell},
    \end{align*}
    where $\xi := H_{\Delta} / \alpha$.
\end{lemma}

We defer the proof of \Cref{lem:framework-hst-det} to \Cref{sec:proof-lem-framework-hst-det} and proceed to finish the proof of \Cref{thm:cost-random-subtree}.

\begin{proof}[Proof of \Cref{thm:cost-random-subtree}]
By \Cref{lem:framework-hst-det}, for $\xi := H_{\Delta} / \alpha$,
\begin{align}\label{eqn:ub-exp-total-cost}
    \cost(\RS; \D, \D)
    \leq 3H_{\Delta} \sum_{j=0}^{h-1} \alpha^{j-h+1} \sum_{\ell = 0}^{j} \xi^{\ell} \sum_{v \in V_j} \E[(\hat{r}(v) - \hat{s}(v))^+].
\end{align}
Next, we upper bound the expected number of excess requests for each height $j$.

For every node $v$, since $\hat{r}(v)$ and $\hat{s}(v)$ are identically distributed,
\begin{align}
    \E[(\hat{r}(v) - \hat{s}(v))^+]
    \leq \E[|\hat{r}(v) - \E\hat{r}(v)|] + \E[|\hat{s}(v) - \E\hat{s}(v)|]
    = 2\E[|\hat{r}(v) - \E\hat{r}(v)|]
    \leq 2\cdot \std(\hat{r}(v)),
    \label{eqn:decom-excess-req}
\end{align}
where the last inequality holds by Jensen's inequality.
Moreover,
\begin{align}
    \std(\hat{r}(v))
    = \std(\pb(\mu_{\D_1}(v), \ldots, \mu_{\D_n}(v)))
    \leq \sqrt{\sum_{i=1}^n \mu_{\D_i}(v)}
    = \sqrt{\mu_{\D}(v)}.
    \label{eqn:ub-mbd-std}
\end{align}
As a result, for each $j \in \{0, 1, \ldots, h - 1\}$,
\begin{align}\label{eqn:ub-excess-req-per-level}
    \sum_{v \in V_j} \E[(\hat{r}(v) - \hat{s}(v))^+]
    \leq 2\sum_{v \in V_j} \sqrt{\mu_{\D}(v)}
    \leq 2\sqrt{n|V_j|}
    = 2\sqrt{n \Delta^{h-j}},
\end{align}
where the first inequality follows from \eqref{eqn:decom-excess-req} and \eqref{eqn:ub-mbd-std}, and the second inequality holds by the concavity of $f(t) = \sqrt{t}$ and the fact that $\sum_{v \in V_j} \mu_{\D}(v) = n$.

Finally, combining \eqref{eqn:ub-exp-total-cost} and \eqref{eqn:ub-excess-req-per-level} concludes the proof.
\end{proof}

\subsection{Proof of \Cref{lem:framework-hst-det}}
\label{sec:proof-lem-framework-hst-det}

The proof of \Cref{lem:framework-hst-det} largely follows the proof strategy of \cite[Theorem 4.1]{DBLP:conf/icalp/GuptaL12}, with two key differences.
First, their analysis assumes $\alpha = \Omega(H_{\Delta})$, ensuring the HST is sufficiently well-separated so that the costs incurred at lower levels are dominated by those at higher levels.
In contrast, we do not rely on this separation, and hence the costs incurred at lower levels must be handled explicitly.
Second, our proof is relatively simpler: since we only bound the absolute cost of the algorithm, we avoid the more involved step of characterizing the offline optimum, which is required in \cite{DBLP:conf/icalp/GuptaL12}.

    Recall that the length of each root-edge is $1$, so the length of each root-leaf path is $1 + \beta$, where $\beta \leq 1 / (\alpha - 1) \leq 1$.
    In the proof, instead of assuming all requests lie only at the leaf nodes, we allow some requests to lie at the root, and the matching process of these requests by the $\RS$ algorithm follows the same description.
    We also permit the number of servers to exceed the number of requests.
    For a server set $S$, leaf requests $R$, and root requests $R'$ with $|S| \geq |R| + |R'|$, let $\cost(\RS; S, R \cup R')$ denote the expected cost of the $\RS$ algorithm.
    We will prove the following stronger statement:
    \begin{align*}
        \cost(\RS; S, R \cup R')
        \leq 3H_{\Delta} \left( |R'| \sum_{j=0}^{h - 1} \xi^j + \sum_{j=0}^{h-1} \sum_{v \in V_j} \frac{(\hat{r}(v) - \hat{s}(v))^+}{\alpha^{h - j - 1}} \sum_{\ell = 0}^{j} \xi^{\ell} \right),
    \end{align*}
    which gives the desired bound by setting $R' = \emptyset$.

    We use $v_r$ to denote the root, and let $\gamma$ be the expected number of requests in $R \cup R'$ that traverse a root-edge in the matching produced by the $\RS$ algorithm.
    The following lemma upper bounds $\gamma$ in terms of $|R'|$ and the number of excess requests in the subtrees of $v_r$.

    \begin{lemma}[Lemma 4.3 in \cite{DBLP:conf/icalp/GuptaL12}]\label{lem:root-edge}
        It holds that
        \begin{align*}
            \gamma
            \leq H_{\Delta} \left( |R'| + \sum_{i=1}^{\Delta} (\hat{r}(c_i(v_r)) - \hat{s}(c_i(v_r)))^+ \right).
        \end{align*}
    \end{lemma}

    We analyze the cost of the $\RS$ algorithm by induction on $h$.
    We start from the base case where the HST is a star with $h = 1$.
    Since $\cost(\RS; S, R \cup R') \leq 2\gamma$, by \Cref{lem:root-edge},
    \begin{align*}
        \cost(\RS; S, R \cup R')
        \leq 2H_{\Delta} \left( |R'| + \sum_{i=1}^{\Delta} (\hat{r}(c_i(v_r)) - \hat{s}(c_i(v_r)))^+ \right),
    \end{align*}
    concluding the proof for the base case.

    Now, we assume that $h \geq 2$.
    For $i \in [\Delta]$, let $T_i$ be the $i$-th subtree of $v_r$, rooted at $c_i(v_r)$.
    Let $S_i$ and $R_i$ denote the servers and requests contained in $T_i$, respectively.
    Note that $S = \bigcup_{i=1}^{\Delta} S_i$ and $R = \bigcup_{i=1}^{\Delta} R_i$.
    Let $M_i$ be the set of requests outside $T_i$ that the $\RS$ algorithm matches to servers in $T_i$.
    Then, $R_i \cup M_i$ are all the requests that are either contained in $T_i$ or matched to servers within $T_i$.
    Let $k := \min\{|S_i|, |R_i \cup M_i|\}.$ Order the requests in $R_i \cup M_i$ by their arrival time, and let $X_i$ be the first $k$ of them. Since $T_i$ contains $|S_i|$ servers and the algorithm never bypasses an available server, the requests matched within $T_i$ are exactly those in $X_i$.
    Define $\widehat{R}_i := X_i \cap R_i$.
    Note that $M_i \subseteq X_i$, and $M_i$, $X_i$, and $\widehat{R}_i$ are random variables.

    We recall the following useful decomposition of $\cost(\RS; S, R \cup R')$ from \cite{DBLP:conf/icalp/GuptaL12}.

    \begin{lemma}[Fact 4.8 in \cite{DBLP:conf/icalp/GuptaL12}]\label{lem:fact-4.8}
        It holds that
        \begin{align*}
            \cost(\RS; S, R \cup R')
            \leq \sum_{i=1}^{\Delta} \E[\cost(\RS; S_i, \widehat{R}_i \cup M_i)] + \sum_{i=1}^{\Delta} \E[|M_i|] \cdot (2 + \beta),
        \end{align*}
        where the expectations are taken over the randomness of the $\RS$ algorithm.
    \end{lemma}

    For each $i \in [\Delta]$, let $\eta_i := (\hat{r}(c_i(v_r)) - \hat{s}(c_i(v_r)))^+$ be the number of excess requests in $T_i$.
    By \Cref{lem:root-edge},
    \begin{align}\label{eqn:sum-EMi-ub}
        \sum_{i=1}^{\Delta} \E[|M_i|]
        = \gamma
        \leq H_{\Delta} \left( |R'| + \sum_{i=1}^{\Delta} \eta_i \right).
    \end{align}
    By \Cref{lem:fact-4.8} and \eqref{eqn:sum-EMi-ub},
    \begin{align}
        \cost(\RS; S, R \cup R')
        &\leq \sum_{i=1}^{\Delta} \E[\cost(\RS; S_i, \widehat{R}_i \cup M_i)] + \sum_{i=1}^{\Delta} \E[|M_i|] \cdot (2 + \beta) \notag \\
        &\leq \sum_{i=1}^{\Delta} \E[\cost(\RS; S_i, \widehat{R}_i \cup M_i)] + 3H_{\Delta} \left( |R'| + \sum_{i=1}^{\Delta} \eta_i \right). \label{eqn:upper-cost-rs-s-r-rprimr}
    \end{align}
    Next, we upper bound the first term in \eqref{eqn:upper-cost-rs-s-r-rprimr}.
    By the inductive hypothesis, and since the length of each root-edge of each subtree $T_i$ equals $1 / \alpha$ and $\widehat{R}_i \subseteq R_i$,
    \begin{align*}
        \sum_{i=1}^{\Delta} \E[\cost(\RS; S_i, \widehat{R}_i \cup M_i)]
        &\leq \frac{1}{\alpha} \cdot 3H_{\Delta} \left( \sum_{i=1}^{\Delta} \E[|M_i|] \sum_{j=0}^{h - 2} \xi^j + \sum_{j=0}^{h-2} \sum_{v \in V_j} \frac{(\hat{r}(v) - \hat{s}(v))^+}{\alpha^{h - j - 2}} \sum_{\ell = 0}^{j} \xi^{\ell} \right) \\
        &\leq 3H_{\Delta} \left( \left( |R'| + \sum_{i=1}^{\Delta} \eta_i \right) \sum_{j=1}^{h - 1} \xi^j + \sum_{j=0}^{h-2} \sum_{v \in V_j} \frac{(\hat{r}(v) - \hat{s}(v))^+}{\alpha^{h - j - 1}} \sum_{\ell = 0}^{j} \xi^{\ell} \right),
    \end{align*}
    where the second inequality holds by \eqref{eqn:sum-EMi-ub}.
    Combining the above two displayed equations, we get
    \begin{align*}
        \cost(\RS; S, R \cup R')
        &\leq 3H_{\Delta} \left( \left( |R'| + \sum_{i=1}^{\Delta} \eta_i \right) \sum_{j=0}^{h - 1} \xi^j + \sum_{j=0}^{h-2} \sum_{v \in V_j} \frac{(\hat{r}(v) - \hat{s}(v))^+}{\alpha^{h - j - 1}} \sum_{\ell = 0}^{j} \xi^{\ell} \right) \\
        &= 3H_{\Delta} \left( |R'| \sum_{j=0}^{h - 1} \xi^j + \sum_{j=0}^{h-1} \sum_{v \in V_j} \frac{(\hat{r}(v) - \hat{s}(v))^+}{\alpha^{h - j - 1}} \sum_{\ell = 0}^{j} \xi^{\ell} \right),
    \end{align*}
    where the equality holds by the definition of $\eta_i$'s, concluding the induction.
\section{BBGN Algorithm}

In this section, we turn from the $\RS$ algorithm to the algorithm proposed by \cite{DBLP:journals/algorithmica/BansalBGN14}, henceforth referred to as the BBGN algorithm, which outperforms the $\RS$ algorithm in certain parameter regimes.
The BBGN algorithm is $O(\log n)$-competitive for $\alpha$-HST metrics with any constant $\alpha > 1$ when both $S$ and $R$ are adversarial. Our interest, however, lies in analyzing its cost when both $S$ and $R$ are drawn from $\D$.
As before, we will later translate these guarantees to the setting where servers are adversarial and requests are drawn from $\D$ by applying \Cref{lem:stochastic-to-semi}.

We now state the main theorem of this section.

\begin{theorem}\label{thm:cost-bbgn}
    For any $\Delta$-ary $\alpha$-HST with height $h$ and constant $\alpha \geq 2$ such that $\Delta^{h - 1} \leq n/2$, and for any $\D = \prod_{i=1}^n \D_i$,
    \begin{align*}
        \cost(\BBGN; \D, \D)
        \leq O \left( \sqrt{n} \sum_{j=1}^{h} \frac{\sqrt{\Delta^{h-j+1}}}{\alpha^{h-j}} \cdot \log\frac{n}{\Delta^{h-j}} \right).
    \end{align*}
\end{theorem}

The rest of this section is devoted to proving \Cref{thm:cost-bbgn}.
For each leaf node $v$ of the HST and $\ell \in \{0, 1, \ldots, h\}$, let $p(v, \ell)$ be the ancestor of $v$ with height $\ell$.
The $\BBGN$ algorithm was originally developed in a \emph{restricted reassignment} online model, which permits limited reassignment of previously arrived requests, and then adapted to the standard online model.
We recall here a bound for the cost of the $\BBGN$ algorithm that applies to all $S$ and $R$, obtained by combining \cite[Lemmas 4.2, 4.3, and 4.4]{DBLP:journals/algorithmica/BansalBGN14}, and we refer to \cite{DBLP:journals/algorithmica/BansalBGN14} for a more formal description of the algorithm.

\begin{lemma}[\cite{DBLP:journals/algorithmica/BansalBGN14}]\label{lem:BBGN-cost-condition}
    Consider an $\alpha$-HST with $\alpha \geq 2$ being a constant.
    For all $S$ and $R$, in addition to the resulting (random) matching $M$ between $S$ and $R$, the $\BBGN$ algorithm also produces a (random) min-cost perfect matching $M_{\opt}$ between $S$ and $R$ such that
    \begin{align*}
        \E \left[ \sum_{r \in R} \delta(M(r), r) \; \Bigg \lvert \; M_{\opt} \right]
        \leq O \left( \sum_{r \in R} \sum_{\ell = 1}^{L(r)} \alpha^{\ell - h} \cdot \log \hat{s}(p(r, \ell)) \right),
    \end{align*}
    where $L(r)$ is the height of the least common ancestor of $r$ and $M_{\opt}(r)$.
\end{lemma}

It is well known that min-cost perfect matchings in HSTs admit a useful characterization (see, e.g., \cite[Lemma 4.1]{DBLP:journals/algorithmica/BansalBGN14}): every min-cost perfect matching between $S$ and $R$ can be obtained by a greedy algorithm that repeatedly matches the closest request-server pair and then recurses on the remaining instance.
Equivalently, the number of matched pairs whose two endpoints lie in different subtrees of a given internal node can be expressed as follows.

\begin{fact}\label{fact:hst-opt-sol}
    For any HST, in any min-cost perfect matching between $S$ and $R$, the number of matches whose two endpoints lie in different subtrees of an internal node $v$ equals
    \begin{align*}
        \min \left\{ \sum_{i=1}^{\Delta} (\hat{r}(c_i(v)) - \hat{s}(c_i(v)))^+, \sum_{i=1}^{\Delta} (\hat{s}(c_i(v)) - \hat{r}(c_i(v)))^+ \right\}.
    \end{align*}
\end{fact}

Combining \Cref{lem:BBGN-cost-condition} and \Cref{fact:hst-opt-sol}, we can then upper bound the expected cost of the $\BBGN$ algorithm as in the following lemma.

\begin{lemma}\label{lem:ub-bbgn-worst}
    For any $\alpha$-HST with $\alpha \geq 2$ being a constant, and for all $S$ and $R$,
    \begin{align*}
        \cost(\BBGN; S, R)
        \leq O\left( \sum_{j=1}^{h} \alpha^{j-h} \sum_{v \in V_j} \log(\hat{s}(v) + e)  \sum_{i=1}^{\Delta} |\hat{r}(c_i(v)) - \hat{s}(c_i(v))| \right).
    \end{align*}
\end{lemma}

\begin{proof}
    By \Cref{lem:BBGN-cost-condition}, the $\BBGN$ algorithm produces an additional min-cost perfect matching $M_{\opt}$ between $S$ and $R$.
    For each $r \in R$, let $L(r)$ be the height of the least common ancestor of $r$ and $M_{\opt}(r)$.
    For each internal node $v$, let $q(v)$ be the number of requests $r$ such that $p(r, L(r)) = v$, which also equals the number of matches in $M_{\opt}$ whose two endpoints lie in different subtrees of $v$.
    By \Cref{fact:hst-opt-sol},
    \begin{align} \label{eqn:ub-q(v)}
        q(v)
        &= \min \left\{ \sum_{i=1}^{\Delta} (\hat{r}(c_i(v)) - \hat{s}(c_i(v)))^+, \sum_{i=1}^{\Delta} (\hat{s}(c_i(v)) - \hat{r}(c_i(v)))^+ \right\} \notag \\
        &\leq \sum_{i=1}^{\Delta} |\hat{r}(c_i(v)) - \hat{s}(c_i(v))|.
    \end{align}
    If we denote $M$ as the (random) matching resulting from the $\BBGN$ algorithm, then by \Cref{lem:BBGN-cost-condition},
    \begin{align}
        \E \left[ \sum_{r \in R} \delta(M(r), r) \; \Bigg \lvert \; M_{\opt} \right]
        &\leq O \left( \sum_{r \in R} \sum_{\ell = 1}^{L(r)} \alpha^{\ell - h} \cdot \log \hat{s}(p(r, \ell)) \right) \notag \\
        &\leq O \left( \sum_{r \in R} \sum_{\ell = 1}^{L(r)} \alpha^{\ell - h} \cdot \log \hat{s}(p(r, L(r))) \right) \notag \\
        &\leq O \left( \sum_{r \in R} \indc{L(r) \geq 1} \cdot \alpha^{L(r) - h} \cdot \log \hat{s}(p(r, L(r))) \right) \notag \\
        &\leq O \left( \sum_{j=1}^h \alpha^{j-h} \sum_{v \in V_j} \log(\hat{s}(v) + e) \cdot q(v) \right) \notag \\
        &\leq O \left( \sum_{j=1}^{h} \alpha^{j-h} \sum_{v \in V_j} \log (\hat{s}(v) + e) \sum_{i=1}^{\Delta} |\hat{r}(c_i(v)) - \hat{s}(c_i(v))| \right), \label{eqn:ub-final-match-bbgn}
    \end{align}
    where the third inequality holds since $\alpha \geq 2$, and the last inequality holds by \eqref{eqn:ub-q(v)}.
    Since \eqref{eqn:ub-final-match-bbgn} does not depend on $M_{\opt}$, it is also an upper bound for $\cost(\BBGN; S, R)$, concluding the proof.
\end{proof}

Now, we are ready to finish the proof of \Cref{thm:cost-bbgn}.

\begin{proof}[Proof of \Cref{thm:cost-bbgn}]
    By \Cref{lem:ub-bbgn-worst},
    \begin{align}
        &\cost(\BBGN; \D, \D) \notag \\
        &\leq O\left( \sum_{j=1}^{h} \alpha^{j-h} \sum_{v \in V_j}  \sum_{i=1}^{\Delta} \E[\log(\hat{s}(v) + e) \cdot |\hat{r}(c_i(v)) - \hat{s}(c_i(v))|] \right) \notag \\
        &\leq O\left( \sum_{j=1}^h \alpha^{j-h} \sum_{v \in V_j} \sqrt{\E[\log^2 (\hat{s}(v) + e)]} \sum_{i=1}^{\Delta} \sqrt{\E[(\hat{r}(c_i(v)) - \hat{s}(c_i(v)))^2]} \right), \label{eqn:cost-bbgn-init}
    \end{align}
    where the second inequality holds by the Cauchy-Schwarz inequality.

    We first bound the last summation in \eqref{eqn:cost-bbgn-init}.
    For all internal node $v$ and $i \in [\Delta]$,
    \begin{align*}
        \sqrt{\E[(\hat{r}(c_i(v)) - \hat{s}(c_i(v)))^2]}
        &\leq \sqrt{2\E[(\hat{r}(c_i(v)) - \mu_{\D}(c_i(v)))^2 + (\hat{s}(c_i(v)) - \mu_{\D}(c_i(v)))^2]} \\
        &= 2 \cdot \std(\hat{r}(c_i(v)))
        \leq 2\sqrt{\mu_{\D}(c_i(v))},
    \end{align*}
    where the first inequality follows since $(x + y)^2 \leq 2(x^2 + y^2)$, the equality holds since $\hat{r}(c_i(v))$ and $\hat{s}(c_i(v))$ are identically distributed with $\E[\hat{r}(c_i(v))] = \E[\hat{s}(c_i(v))] = \mu_{\D}(c_i(v))$, and the last inequality holds by \eqref{eqn:ub-mbd-std}.
    It follows that
    \begin{align}\label{eqn:bbgn-ub-sum-std}
        \sum_{i=1}^{\Delta} \sqrt{\E[(\hat{r}(c_i(v)) - \hat{s}(c_i(v)))^2]}
        \leq 2\sum_{i=1}^{\Delta} \sqrt{\mu_{\D}(c_i(v))}
        \leq 2\sqrt{\Delta \cdot \mu_{\D}(v)},
    \end{align}
    where the last inequality holds by the concavity of $f(t) = \sqrt{t}$ and the fact that $\sum_{i=1}^{\Delta} \mu_{\D}(c_i(v)) = \mu_{\D}(v)$.

    Next, we bound $\E[\log^2(\hat{s}(v) + e)]$ for every internal node $v$.
    Since $f(t) = \log^2 (t + e)$ is concave over $[0, +\infty)$, by Jensen's inequality,
    \begin{align}\label{eqn:bbgn-ub-log2}
        \E[\log^2(\hat{s}(v) + e)]
        \leq \log^2(\E[\hat{s}(v)] + e)
        = \log^2(\mu_{\D}(v) + e).
    \end{align}

    Finally, combining \eqref{eqn:cost-bbgn-init}, \eqref{eqn:bbgn-ub-sum-std}, and \eqref{eqn:bbgn-ub-log2},
    \begin{align*}
        \cost(\BBGN; \D, \D)
        &\leq O \left( \sqrt{\Delta} \sum_{j=1}^h \alpha^{j-h} \sum_{v \in V_j} \log(\mu_{\D}(v) + e) \sqrt{\mu_{\D}(v)} \right) \\
        &\leq O\left( \sqrt{\Delta} \sum_{j=1}^h \alpha^{j - h} \cdot |V_j| \cdot \sqrt{\frac{n}{|V_j|}} \log \left(\frac{n}{|V_j|} + e\right) \right) \\
        &\leq O\left( \sqrt{n\Delta} \sum_{j=1}^h \alpha^{j - h} \sqrt{\Delta^{h-j}} \log \frac{n}{\Delta^{h-j}} \right),
    \end{align*}
    where the second inequality holds by the concavity of $f(t) = \sqrt{t} \log (t + e)$ over $[0, +\infty)$ and the fact that $\sum_{v \in V_j} \mu_{\D}(v) = n$, and the third inequality holds since $|V_j| = \Delta^{h - j} \leq n / 2$ for every $j \in [h]$.
    This concludes the proof.
\end{proof}
\section{Euclidean Metrics}
\label{sec:euc}

In this section, we turn to the Euclidean setting, where $\calX = [0, 1]^d$ and $\delta$ is the Euclidean distance. This case is widely studied in online metric matching.
We begin by stating our paper's main theorem, which characterizes the competitive ratio achievable under $\sigma$-smooth request distributions.
Recall that a measure $\mu$ over $[0, 1]^d$ is $\sigma$-smooth for $\sigma \in (0, 1]$ if $\mu(\calX') \leq \U(\calX') / \sigma$ for every measurable subset $\calX' \subseteq [0, 1]^d$, where $\U$ is the uniform distribution over $[0, 1]^d$.

\begin{theorem}[Main Theorem]\label{thm:cp-euc}
    For $[0, 1]^d$ with the Euclidean distance, suppose that $S$ is adversarial and $R$ is drawn from $\D = \prod_{i=1}^n \D_i$, where each $\D_i$ is $\sigma$-smooth for $\sigma \in (0, 1]$.
    Moreover, suppose that the algorithm is given one sample from each $\D_i$.
    Then, there exists an algorithm $\calA_{\RS}$ with competitive ratio
    \begin{align*}
        \begin{cases}
            O(\sigma^{-1}), & d = 1,\\
            O(\sigma^{-\frac{1}{2}}n^{\frac{1}{2} - \frac{1}{2\log(25/12)}}), & d = 2,\\
            O(d^{\frac{3}{2}} \sigma^{-\frac{1}{d}}), & d \geq 3,
        \end{cases}
    \end{align*}
    and an algorithm $\calA_{\BBGN}$ with competitive ratio
    \begin{align*}
        \begin{cases}
            O(\sigma^{-1}\log n), & d = 1,\\
            O(\sigma^{-\frac{1}{2}}\log^2 n), & d = 2,\\
            O(d^{\frac{3}{2}} \sigma^{-\frac{1}{d}}), & d \geq 3.
        \end{cases}
    \end{align*}
    In particular, both algorithms do not need to know the correspondence between distributions and samples.
\end{theorem}

In particular, these results show that in one dimension the dependence on $n$ is either logarithmic or absent, in two dimension the competitive ratio grows sublinearly in $n$, and in higher dimensions the bound is dimension-dependent but independent of $n$.

The proof of \Cref{thm:cp-euc} consists of three ingredients.
First, \Cref{lem:stochastic-to-semi} shows that, given sample access to $\D$, the adversarial-server setting can be reduced to the stochastic case where the servers are also drawn from $\D$.
Second, \Cref{lem:euc-alg} provides two algorithms for Euclidean metrics along with upper bounds on their expected costs. Finally, \Cref{lem:lb-smooth} establishes lower bounds on the offline optimum.
We present the latter two results next, and then combine all three ingredients to complete the proof.

\begin{lemma}\label{lem:euc-alg}
    For $[0, 1]^d$ with the Euclidean distance, and for any $\D = \prod_{i=1}^n \D_i$, there exists an algorithm $\calA_{\RS}$ such that
    \begin{align*}
        \cost(\calA_{\RS}; \D, \D)
        \leq \begin{cases}
            O(\sqrt{n}), & d = 1,\\
            O\left(n^{1 - \frac{1}{2 \log(25/12)}}\right), & d = 2,\\
            O(d^{\frac{3}{2}} n^{1-\frac{1}{d}}), & d \geq 3,
        \end{cases}
    \end{align*}
    and an algorithm $\calA_{\BBGN}$ such that
    \begin{align*}
        \cost(\calA_{\BBGN}; \D, \D)
        \leq
        \begin{cases}
            O(\sqrt{n} \log n), & d = 1,\\
            O(\sqrt{n} \log^2 n), & d = 2,\\
            O(d^{\frac{3}{2}}n^{1-\frac{1}{d}}), & d \geq 3.
        \end{cases}
    \end{align*}
\end{lemma}

We defer the proof of \Cref{lem:euc-alg} to \Cref{sec:proof-lem-euc-alg} and give lower bounds for the offline optimum in the following lemma.

\begin{lemma}\label{lem:lb-smooth}
    For $[0, 1]^d$ with the Euclidean distance, and for all $S$ and $\D = \prod_{i=1}^n \D_i$ such that each $\D_i$ is $\sigma$-smooth for $\sigma \in (0, 1]$,
    \begin{align*}
        \opt(S, \D)
        \geq \begin{cases}
            \Omega(\sigma \sqrt{n}), & d = 1,\\
            \Omega(\sigma^{\frac{1}{d}} n^{1 - \frac{1}{d}}), & d \geq 2.
        \end{cases}
    \end{align*}
\end{lemma}

We defer the proof of \Cref{lem:lb-smooth} to \Cref{sec:proof-lem-lb-smooth}.
Now, we have collected all the necessary ingredients to finish the proof of \Cref{thm:cp-euc}.

\begin{proof}[Proof of \Cref{thm:cp-euc}]
    By \Cref{lem:stochastic-to-semi}, given an algorithm $\calA$, which is provided with a sample from each $\D_i$, there exists an algorithm $\calA'$, which does not need to know the correspondence between distributions and samples, such that $\cost(\calA'; S, \D) \leq \opt(S, \D) + \cost(\calA; \D, \D)$.
    Hence, it suffices to give an algorithm $\calA$ such that $\cost(\calA; \D, \D) / \opt(S, \D)$ is upper bounded by the desired competitive ratio.
    The theorem then follows by applying the algorithms $\calA_{\RS}$ and $\calA_{\BBGN}$ given in \Cref{lem:euc-alg}, and the lower bounds for $\opt(S, \D)$ given in \Cref{lem:lb-smooth}.
\end{proof}

\subsection{Proof of \Cref{lem:euc-alg}}
\label{sec:proof-lem-euc-alg}

Define a hierarchical decomposition $\calH_0, \calH_1, \ldots, \calH_h$ of $[0, 1]^d$, where $h$ will be determined later, such that
\begin{align*}
    \calH_i := \left\{\prod_{\ell = 1}^d I(i, \lambda_{\ell}) \; \Big \lvert \; \lambda_1, \ldots, \lambda_d \in [2^{h-i}] \right\}, \quad \text{where} \quad 
    I(i, \lambda) :=
    \begin{cases}
        [2^{i-h}(\lambda-1), 2^{i-h} \lambda), & \lambda < 2^{h-i} ,\\
        [2^{i-h}(\lambda-1), 2^{i-h} \lambda], & \lambda = 2^{h-i}.
    \end{cases}
\end{align*}
In other words, $\calH_i$ is the partition of $[0, 1]^d$ into $2^{d(h-i)}$ subcubes with side-length $2^{i-h}$.
This gives a laminar family\footnote{Recall that a \emph{laminar family} $\calF$ is a family of subsets such that for any $X, Y \in \calF$, one of the three following cases holds: (1) $X \subseteq Y$, (2) $Y \subseteq X$, or (3) $X \cap Y = \emptyset$.} $\calH := \calH_0 \cup \calH_1 \cup \ldots \cup \calH_h$.
See \Cref{fig:hg} for an illustration with $d = 2$.
To prove \Cref{lem:euc-alg}, we first construct a $2^d$-ary $2$-HST with height $h$ from $\calH$, and then we show that it suffices to upper bound the cost of an algorithm for the resulting HST metric, which enables us to apply our algorithmic results for HST metrics.

\begin{figure}
    \centering
    \includegraphics[width=0.8\linewidth]{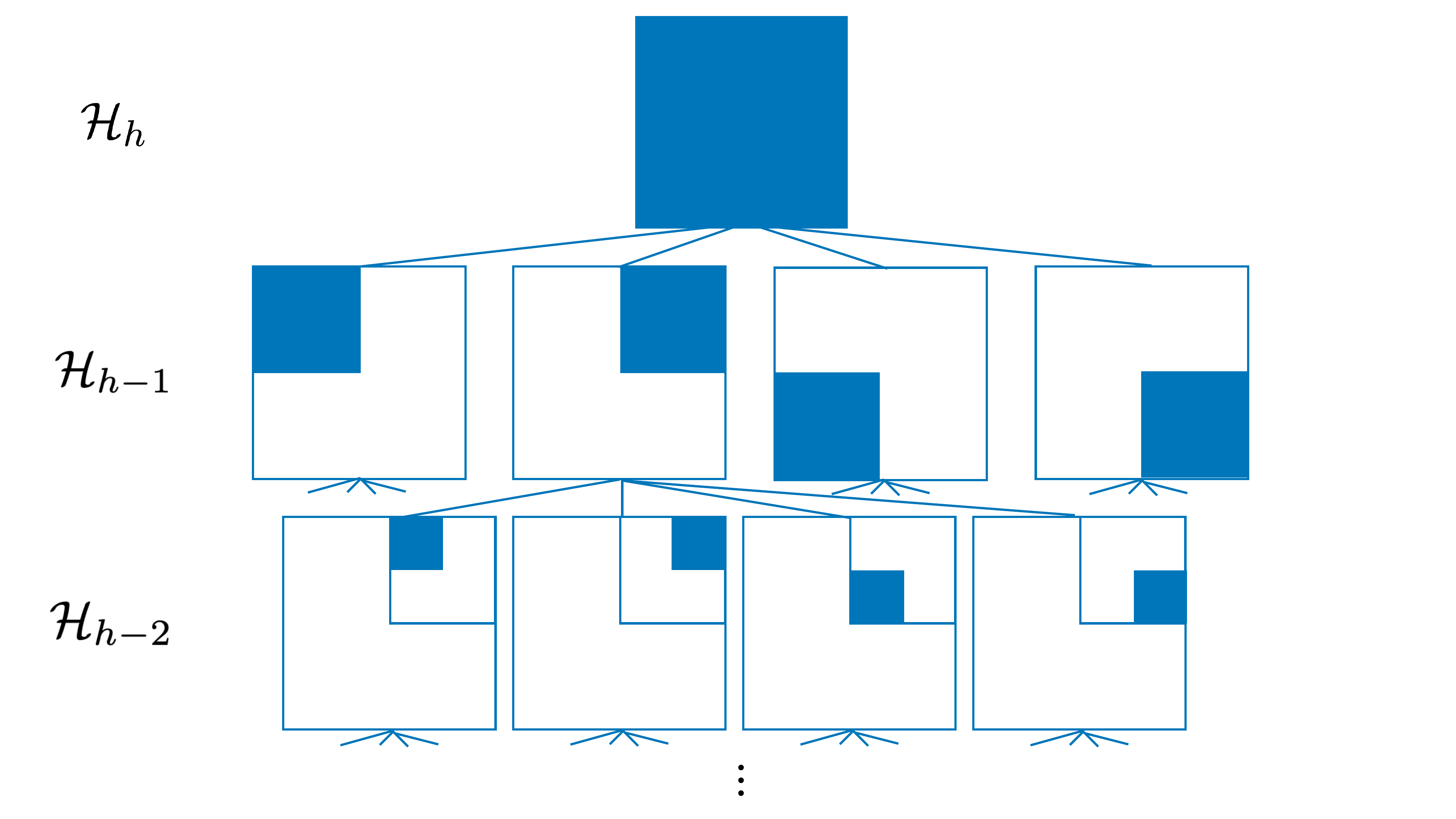}
    \caption{Hierarchical decomposition for $[0, 1]^2$.} 
    \label{fig:hg}
\end{figure}

We construct a $2^d$-ary $2$-HST with height $h$, denoted as $\calT$, from $\calH$ as follows: each cube in $\calH$ corresponds to a node, and the children of a node corresponding to $H \in \calH$ are the nodes corresponding to maximal subsets of $H$ in $\calH$.
For every $x \in [0, 1]^d$, let $\calT(x)$ be the leaf node corresponding to the (unique) cube in $\calH_0$ that contains $x$.
For $S \subseteq [0, 1]^d$, define $\calT(S) := \{\calT(s) \mid s \in S\}$.
We show in the following lemma that a cost upper bound for an algorithm on $\calT$ gives rise to a cost upper bound for the corresponding algorithm on $[0, 1]^d$.

\begin{lemma}\label{lem:hg-to-hst}
    Given an algorithm $\calA$ on $\calT$, there exists an algorithm $\calA'$ on $[0, 1]^d$ such that for all $S$ and $R$, $\cost(\calA'; S, R) \leq \sqrt{d} \cdot (\cost(\calA; \calT(S), \calT(R)) / 4 + n \cdot 2^{-h})$.
\end{lemma}

\begin{proof}
    Given the server set $S = \{s_1, \ldots, s_n\}$, the algorithm $\calA'$ initializes the given algorithm $\calA$ with the server set being $\calT(S)$.
    For each arriving request $r \in [0, 1]^d$, if $\calA$ matches request $\calT(r)$ to server $\calT(s)$ for $s \in S$, then $\calA'$ matches $r$ to $s$.
    To establish the desired upper bound for $\cost(\calA'; S, R)$, it suffices to show that $\norm{s - r}_2 \leq \sqrt{d} \cdot (\delta(\calT(s), \calT(r)) / 4 + 2^{-h})$ for all $s, r \in [0, 1]^d$, where $\delta(v, v')$ denotes the distance between $v$ and $v'$ on $\calT$.

    Fix $s, r \in [0, 1]^d$.
    Let $k$ be the smallest integer in $\{0, 1, \ldots, h\}$ such that there exists $H \in \calH_k$ that contains both $s$ and $r$.
    In other words, $H$ is the (unique) smallest cube in $\calH$ that contains both $s$ and $r$, which implies $\norm{s - r}_2 \leq \diam(H) = \sqrt{d} \cdot 2^{k - h}$.
    Note that the node corresponding to $H$ is the least common ancestor of $\calT(s)$ and $\calT(r)$, and the height of this node is $k$.
    Therefore,
    \begin{align*}
        \delta(\calT(s), \calT(r))
        = 2\sum_{i=h-k}^{h-1} 2^{-i}
        = 2^{k-h+2} - 2^{2-h}
        \geq \frac{4\norm{s - r}_2}{\sqrt{d}} - 2^{2-h},
    \end{align*}
    concluding the proof.
\end{proof}

By \Cref{lem:hg-to-hst}, it suffices to provide algorithms for any $2^d$-ary $2$-HST with a certain height $h$, for which we apply Theorems~\ref{thm:cost-random-subtree} and~\ref{thm:cost-bbgn}, where the height $h$ is chosen to minimize the resulting cost for $[0, 1]^d$.
We establish upper bounds for the expected cost of the $\RS$ and $\BBGN$ algorithms for the specified HST in the following two corollaries, whose proofs only involve mechanical calculation and are deferred to Appendices~\ref{sec:proof-cor-RS-special-hst} and~\ref{sec:proof-cor-bbgn-specific-hst}, respectively.

\begin{corollary}\label{cor:RS-special-hst}
    For the $2^d$-ary $2$-HST with height
    \begin{align*}
        h := \begin{cases}
            \left \lfloor \frac{\log n}{2 \log(25/12)} \right \rfloor, & d = 2,\\
            \lfloor \log (n) / d \rfloor, & d \neq 2,
        \end{cases}
    \end{align*}
    and for any $\D = \prod_{i=1}^n \D_i$,
    \begin{align*}
        \cost(\RS; \D, \D)
        \leq 
        \begin{cases}
            O(\sqrt{n}), & d = 1,\\
            O\left(n^{1 - \frac{1}{2 \log(25/12)}}\right), & d = 2,\\
            O(dn^{1-\frac{1}{d}}), & d \geq 3.
        \end{cases}
    \end{align*}
\end{corollary}

\begin{corollary}\label{cor:bbgn-specific-hst}
    For the $2^d$-ary $2$-HST with height $h := \lfloor \log (n) / d \rfloor$, and for any $\D = \prod_{i=1}^n \D_i$,
    \begin{align*}
        \cost(\BBGN; \D, \D)
        \leq
        \begin{cases}
            O(\sqrt{n} \log n), & d = 1,\\
            O(\sqrt{n} \log^2 n), & d = 2,\\
            O(dn^{1-\frac{1}{d}}), & d \geq 3.
        \end{cases}
    \end{align*}
\end{corollary}

For $\D = \prod_{i=1}^n \D_i$, where each $\D_i$ is supported on $[0, 1]^d$, denote $\calT(\D)$ as the distribution followed by $\calT(S)$, where $S \sim \D$.
Denote the algorithm for $[0, 1]^d$ obtained by applying \Cref{lem:hg-to-hst} to the $\RS$ algorithm as $\calA_{\RS}$, which implies $\cost(\calA_{\RS}; \D, \D) \leq \sqrt{d} \cdot (\cost(\RS; \calT(\D), \calT(\D)) / 4 + n \cdot 2^{-h})$ for any $\D = \prod_{i=1}^n \D_i$.
Next, we upper bound the cost of $\calA_{\RS}$ by \Cref{cor:RS-special-hst}.
For $d = 1$,
\begin{align*}
    \cost(\mathcal A_{\RS}; \D, \D) & \leq O( \cost(\RS; \calT(\D), \calT(\D)) + n \cdot 2^{-\floor{\log n}}) \leq O(\sqrt{n}).
\end{align*}
For $d = 2$,
\begin{align*}
    \cost(\mathcal A_{\RS}; \D, \D) \leq O\left( \cost(\RS; \calT(\D), \calT(\D)) + n \cdot 2^{-\floor{\frac{\log n}{2 \log (15/12)}}} \right)
    \leq O\left(n^{1 - \frac{1}{2 \log(25/12)}}\right).
\end{align*}
For $d \geq 3$, 
\begin{align*}
    \cost(\mathcal A_{\RS}; \D, \D) \leq O(\sqrt{d} \cdot (\cost(\RS; \calT(\D), \calT(\D)) + n \cdot 2^{-\floor{\log (n)/d}}))
    \leq O(d^{\frac{3}{2}}n^{1 - \frac{1}{d}}).
\end{align*}

Denote the algorithm for $[0, 1]^d$ obtained by applying \Cref{lem:hg-to-hst} to the $\BBGN$ algorithm as $\calA_{\BBGN}$, which implies $\cost(\calA_{\BBGN}; \D, \D) \leq \sqrt{d} \cdot (\cost(\BBGN; \calT(\D), \calT(\D)) / 4 + n \cdot 2^{-h})$ for any $\D = \prod_{i=1}^n \D_i$.
Next, we upper bound the cost of $\calA_{\BBGN}$ by \Cref{cor:bbgn-specific-hst}.
For $d = 1$,
\begin{align*}
    \cost(\mathcal A_{\BBGN}; \D, \D) \leq O(\cost(\BBGN; \calT(\D), \calT(\D)) + n \cdot 2^{-\floor{\log n}})
    \leq O(\sqrt{n} \log n).
\end{align*}
For $d=2$, 
\begin{align*}
    \cost(\mathcal A_{\BBGN}; \D, \D) \leq O(\cost(\BBGN; \calT(\D), \calT(\D)) + n \cdot 2^{-\floor{\log (n)/2}})
    \leq O(\sqrt{n} \log^2 n).
\end{align*}
For $d \geq 3$,
\begin{align*}
    \cost(\mathcal A_{\BBGN}; \D, \D) \leq O(\sqrt{d} \cdot (\cost(\BBGN; \calT(\D), \calT(\D)) + n \cdot 2^{-\floor{\log (n)/d}}))
    \leq O(d^{\frac{3}{2}} n^{1-\frac{1}{d}}),
\end{align*}
concluding the proof of \Cref{lem:euc-alg}.

\subsection{Proof of \Cref{lem:lb-smooth}}
\label{sec:proof-lem-lb-smooth}

The proof for $d \geq 2$ relies on the following lower bound for the minimum cost of matching each request, which is established by employing nearest-neighbor-distance.

\begin{lemma}[Lemma 14 in \cite{DBLP:journals/ior/YangY26}]\label{lemma:near-neigh-dist}
    Let $\D_i$ be a $\sigma$-smooth distribution over $[0, 1]^d$ with $\sigma \in (0, 1]$.
    For all $d \geq 1$ and finite $S \subseteq [0, 1]^d$, 
    \begin{align*}
        \E_{r_i \sim \D_i} \left[ \min_{s \in S} \norm{s - {r_i}}_2 \right] \geq \Omega \left( \sigma^{\frac{1}{d}}|S|^{-\frac{1}{d}} \right).
    \end{align*}
\end{lemma}

To see that \Cref{lemma:near-neigh-dist} implies the lower bound for $d \geq 2$, note that
\begin{align*}
    \opt(S, \D)
    \geq \sum_{i=1}^n \E_{r_i \sim \D_i} \left[ \min_{s \in S} \norm{s-r_i}_2 \right]
    \geq \Omega \left( \sigma^{\frac{1}{d}} n^{1 - \frac{1}{d}} \right),
\end{align*}
as desired.

The rest of this subsection is devoted to proving the lower bound for $d = 1$.
Fix $S$ and $\D = \prod_{i=1}^n \D_i$, and let $R \sim \D$.
Fix an arbitrary min-cost perfect matching $M$ between $S$ and $R$.
For each $x \in [0, 1]$, define $\hat{m}(x)$ as the number of matches in $M$ that ``cross'' $x$, i.e., one endpoint of the match is in $[0, x]$, and the other endpoint is in $[x, 1]$.
Note that
\begin{align*}
    \opt(S, \D)
    = \E \left[ \int_0^1 \hat{m}(x) \mathrm{d} x \right].
\end{align*}
Let $L := \sigma / 4$.
For each $x \in [0, 1 - L]$, define $\hat{s}(x) := |S \cap [x, x + L]|$ as the number of servers in $[x, x + L]$ and $\hat{r}(x) := |R \cap [x, x + L]|$ as the (random) number of requests in $[x, x + L]$.
Note that if $\hat{s}(x) > \hat{r}(x)$, then at least $\hat{s}(x) - \hat{r}(x)$ servers in $[x, x + L]$ have to be matched to requests outside of $[x, x + L]$; similarly, if $\hat{s}(x) < \hat{r}(x)$, then at least $\hat{r}(x) - \hat{s}(x)$ requests in $[x, x + L]$ have to be matched to servers outside of $[x, x + L]$.
Hence, for each $x \in [0, 1 - L]$, $\hat{m}(x) + \hat{m}(x + L) \geq |\hat{s}(x) - \hat{r}(x)|$.
As a result,
\begin{align*}
    \int_0^{1 - L} |\hat{s}(x) - \hat{r}(x)| \mathrm{d} x
    \leq \int_0^{1 - L} (\hat{m}(x) + \hat{m}(x + L)) \mathrm{d} x
    \leq 2 \int_0^1 \hat{m}(x) \mathrm{d} x.
\end{align*}
Hence, it suffices to show that
\begin{align}\label{eqn:lb-expect-sum-abs-diff}
    \E \left[ \int_0^{1 - L}  |\hat{s}(x) - \hat{r}(x)| \mathrm{d} x \right]
    \geq \Omega(\sigma \sqrt{n}).
\end{align}

Fix $x \in [0, 1 - L]$, and we analyze $\E [|\hat{s}(x) - \hat{r}(x)|]$.
For each $i \in [n]$, we use $\mu_{\D_i}$ to denote the density of $\D_i$ with respect to the uniform distribution over $[0, 1]$.
Note that $\hat{r}(x) \sim \pb(\bfw(x))$, where $w_i(x) := \int_{x}^{x + L} \mu_{\D_i}(y) \mathrm{d} y$ for each $i \in [n]$.
Define $W(x) := \norm{\bfw(x)}_1$, which implies $\E[\hat{r}(x)] = W(x)$.
The following lemma lower bounds $\E[|\hat{s}(x) - \hat{r}(x)|]$.

\begin{lemma}\label{lem:bound-mean-devia}
    $\E[|\hat{s}(x) - \hat{r}(x)|] \geq \Omega(\sqrt{W(x)}) - O(1)$.
\end{lemma}

\begin{proof}
    By the triangle inequality,
    \begin{align*}
        \E[|\hat{r}(x) - W(x)|]
        \leq \E[|\hat{r}(x) - \hat{s}(x)|] + |\hat{s}(x) - W(x)|
        \leq 2\E[|\hat{s}(x) - \hat{r}(x)|],
    \end{align*}
    where the second inequality holds by Jensen's inequality.
    Consequently, it suffices to show that
    \begin{align}\label{eqn:lb-mean-abs-devia-crude}
        \E[|\hat{r}(x) - W(x)|]
        = \E[|\pb(\bfw(x)) - W(x)|]
        \geq \Omega(\sqrt{W(x)}) - O(1).
    \end{align}
    For each $i \in [n]$, by the $\sigma$-smoothness of $\D_i$,
    \begin{align*}
        w_i(x)
        = \int_{x}^{x + L} \mu_{\D_i}(y) \mathrm{d} y
        \leq \frac{L}{\sigma}
        \leq \frac{1}{2}.
    \end{align*}
    Let $\bfw' \in \R^n_{\geq 0}$ satisfy
    \begin{align*}
        w_i' =
        \begin{cases}
            1/2, & i \leq \lfloor 2W(x) \rfloor, \\
            W(x) - \lfloor 2W(x) \rfloor / 2, & i = \lfloor 2W(x) \rfloor + 1, \\
            0, & \text{otherwise},
        \end{cases}
    \end{align*}
    which gives $\norm{\bfw'}_1 = W(x)$ and $\norm{\bfw'}_{\infty} \leq 1/2$.
    Since $\bfw' \succ \bfw(x)$, by \Cref{cor:convex-order-mean-abs-devia},
    \begin{align*}
        \E[|\pb(\bfw(x)) - W(x)|]
        \geq \E[|\pb(\bfw') - W(x)|].
    \end{align*}
    This would imply \eqref{eqn:lb-mean-abs-devia-crude} since
    \begin{align*}
        \E[|\pb(\bfw') - W(x)|]
        &\geq \E\left[ \left| \bin\left(\lfloor 2W(x) \rfloor, \frac{1}{2} \right) - \frac{\lfloor 2W(x) \rfloor}{2} \right| \right] - 1 \\
        &\geq \frac{\sqrt{\lfloor 2W(x) \rfloor}}{2 \sqrt{2}} - 1
        \geq \Omega(\sqrt{W(x)}) - O(1),
    \end{align*}
    where the second inequality holds by the following probabilistic bound.

    \begin{claim}[\cite{berend2013sharp}]\label{cla:binom-pro-bound}
        Let $Z \sim \bin(n, p)$, with $n \geq 2$ and $p \in [1 / n, 1 - 1 / n]$.
        Then, we have
        \begin{align*}
            \E[|Z - \E Z|] \geq \std(Z) / \sqrt{2}.
        \end{align*}
    \end{claim}
\end{proof}

By \Cref{lem:bound-mean-devia},
\begin{align*}
    \E \left[ \int_{0}^{1 - L}  |\hat{s}(x) - \hat{r}(x)| \mathrm{d} x \right]
    = \int_{0}^{1 - L}  \E[|\hat{s}(x) - \hat{r}(x)|] \mathrm{d} x
    \geq \Omega \left( \int_{0}^{1 - L} \sqrt{W(x)} \mathrm{d} x \right) - O(1),
\end{align*}
and \eqref{eqn:lb-expect-sum-abs-diff} follows from the following lemma.

\begin{lemma}
It holds that
\begin{align*}
    \int_{0}^{1 - L} \sqrt{W(x)} \mathrm{d} x
    \geq \Omega(\sigma \sqrt{n}).
\end{align*}
\end{lemma}

\begin{proof}
By the definition of $W(x)$,
\begin{align*}
    \int_{0}^{1 - L} W(x) \mathrm{d} x
    &= \int_{0}^{1 - L} \sum_{i=1}^n \int_{x}^{x + L} \mu_{\D_i}(y) \mathrm{d} y \mathrm{d} x
    = \sum_{i=1}^n \int_{0}^1 \mu_{\D_i}(y) \int_{0}^{1 - L} \indc{y \in [x, x + L]} \mathrm{d} x \mathrm{d} y \\
    &= \sum_{i=1}^n \int_{0}^1 \mu_{\D_i}(y) \cdot \min\{L, y, 1 - y\} \mathrm{d} y
    \geq L \sum_{i=1}^n \int_{L}^{1 - L} \mu_{\D_i}(y) \mathrm{d} y.
\end{align*}
For each $i \in [n]$, by the $\sigma$-smoothness of $\D_i$,
\begin{align*}
    \int_{L}^{1 - L} \mu_{\D_i}(y) \mathrm{d} y
    = 1 - \int_{0}^L \mu_{\D_i}(y) \mathrm{d} y - \int_{1 - L}^1 \mu_{\D_i}(y) \mathrm{d} y
    \geq 1 - \frac{2L}{\sigma}
    = \frac{1}{2}.
\end{align*}
Combining the above two displayed equations,
\begin{align}\label{eqn:lb-int-wx}
    \int_{0}^{1 - L} W(x)
    \geq \frac{Ln}{2}
    = \Omega(\sigma n).
\end{align}
For every $x \in [0, 1 - L]$, since $W(x) \in [0, n]$, we have $\sqrt{W(x)} \geq W(x) / \sqrt{n}$.
It follows that
\begin{align*}
    \int_{0}^{1 - L} \sqrt{W(x)} \mathrm{d} x
    \geq \frac{1}{\sqrt{n}} \int_0^{1 - L} W(x) \mathrm{d} x
    \geq \Omega(\sigma \sqrt{n}),
\end{align*}
where the second inequality holds by \eqref{eqn:lb-int-wx}, concluding the proof.
\end{proof}
\section{Discussion and Future Directions}

In this paper, we study the online metric matching problem for the Euclidean space $[0, 1]^d$ when servers are adversarial and requests are independently drawn from distinct smooth distributions. We present an $O(1)$-competitive algorithm for $[0, 1]^d$ with $d \neq 2$, given a single sample from each request distribution.
A key feature of our approach is that, by directly upper-bounding the algorithm's cost after a simple deterministic metric embedding, we bypass the $\Omega(\log n)$ competitive-ratio barrier that arises in the adversarial setting due to metric distortion.
Since metric embeddings into HSTs have already been proven extremely effective for related online problems such as $k$-server~\cite{DBLP:journals/jacm/BansalBMN15,DBLP:conf/stoc/BubeckCLLM18}, $k$-taxi~\cite{DBLP:conf/soda/GuptaKP24}, and several variants of online metric matching~\cite{DBLP:conf/stoc/EmekKW16,DBLP:conf/soda/BhoreFT24}, a natural and exciting future direction is to extend our techniques to these problems.

Our guarantees rely on requests being independently sampled. An intriguing direction would be to see what forms of correlation among requests might still permit an $o(\log n)$ competitive ratio.
As a starting point, recent breakthroughs in smoothed analysis of online learning~\cite{DBLP:conf/colt/BlockDGR22,DBLP:journals/jacm/HaghtalabRS24} allow each arrival's distribution---while required to be smooth---to depend on the realized history of arrivals and algorithmic decisions, and their techniques may extend to our setting.
In addition, the correlation models studied in online stochastic matching~\cite{DBLP:conf/sigecom/AouadM23} and prophet inequalities~\cite{DBLP:journals/teco/ImmorlicaSW23} present additional promising avenues.

\section*{Acknowledgment}
This work was supported in part by NSF grant CCF-2338226.
We would like to thank the anonymous reviewers for their many helpful comments and suggestions.
Refine.ink was used to check the paper for consistency and clarity.

\bibliographystyle{alpha}
\bibliography{references}

\clearpage
\appendix

\section{Proof of \Cref{lem:stochastic-to-semi}}
\label{sec:proof-lem-stochastic-to-semi}

We first recall the following lemma from \cite{DBLP:journals/ior/YangY26}.

\begin{lemma}[Theorem 1 in \cite{DBLP:journals/ior/YangY26}]\label{lem:thm-1-yy26}
    Given an algorithm $\calA$ and a set $I \subseteq \calX$ with $|I| = n$, there exists an algorithm $\calA'$ such that $\cost(\calA'; S, R) \leq \opt(S, I) + \cost(\calA; I, R)$.
\end{lemma}

Let $I$ be the set of provided samples drawn from request distributions.
By \Cref{lem:thm-1-yy26}, there exists an algorithm $\calA'$ such that
\begin{align*}
    \cost(\calA'; S, R)
    \leq \opt(S, I) + \cost(A; I, R).
\end{align*}
Since $I \sim \D$ and $R \sim \D$ are independent, taking expectations on both sides yields
\begin{align*}
    \cost(\calA; S, \D)
    \leq \opt(S, \D) + \cost(A; \D, \D),
\end{align*}
as desired.
\section{Proof of \Cref{cor:RS-special-hst}}
\label{sec:proof-cor-RS-special-hst}

By \Cref{thm:cost-random-subtree},
\begin{align*}
    \cost(\RS; \D, \D)
    &\leq O \left( d \sqrt{n} \sum_{j=0}^{h-1} \frac{2^{(h-j)d/2}}{2^{h-j-1}} \sum_{\ell = 0}^{j} \left( \frac{H_{2^d}}{2} \right)^{\ell} \right) \\
    &= O \left( d \sqrt{n} \sum_{j=0}^{h-1} 2^{(h-j)(\frac{d}{2}-1)} \cdot \frac{(H_{2^d} / 2)^{j+1}-1}{H_{2^d}/2-1} \right)
\end{align*}
for the $2^d$-ary $2$-HST with height $h$.
When $d = 1$, $H_{2^d} = 3/2$, and
\begin{align*}
    \sum_{j=0}^{h-1} 2^{(h-j)(\frac{d}{2}-1)} \cdot \frac{(H_{2^d} / 2)^{j+1}-1}{H_{2^d}/2-1}
    = O \left( \sum_{j=0}^{h - 1} 2^{-(h-j)/2} \right)
    = O(1).
\end{align*}
When $d = 2$, $H_{2^d} = 25 / 12$, and
\begin{align*}
    \sum_{j=0}^{h-1} 2^{(h-j)(\frac{d}{2}-1)} \cdot \frac{(H_{2^d} / 2)^{j+1}-1}{H_{2^d}/2-1}
    = O \left( \sum_{j=0}^{h-1} \left( \frac{25}{24} \right)^{j+1} \right)
    = O \left( \left( \frac{25}{24} \right)^h \right)
    = O(n^{\frac{1}{2} - \frac{1}{2 \log (25/12)}}).
\end{align*}
When $d \geq 3$, note that $H_{2^d} / 2^{d/2} \leq c$ for some constant $0 < c < 1$, and it follows that
\begin{align*}
    \sum_{j=0}^{h-1} 2^{(h-j)(\frac{d}{2}-1)} \cdot \frac{(H_{2^d} / 2)^{j+1}-1}{H_{2^d}/2-1}
    &= O \left( \sum_{j=0}^{h-1} 2^{(h-j)(\frac{d}{2}-1)} \cdot \left( \frac{H_{2^d}}{2} \right)^j \right) \\
    &= O \left( 2^{h(\frac{d}{2}-1)} \sum_{j=0}^{h-1} \left( \frac{H_{2^d}}{2^{d/2}} \right)^j \right)
    = O \left( 2^{h(\frac{d}{2}-1)} \right)
    = O(n^{\frac{1}{2} - \frac{1}{d}}),
\end{align*}
as desired.
\section{Proof of \Cref{cor:bbgn-specific-hst}}
\label{sec:proof-cor-bbgn-specific-hst}

Note that
\begin{align*}
    2^{d(h-1)}
    \leq n \cdot 2^{-d}
    \leq \frac{n}{2},
\end{align*}
as required by \Cref{thm:cost-bbgn}.
By \Cref{thm:cost-bbgn},
\begin{align*}
    \cost(\BBGN; \D, \D)
    &\leq O \left( \sqrt{n} \sum_{j=1}^{h} \frac{2^{(h-j+1)d/2}}{2^{h-j}} \cdot \log\frac{n}{2^{(h-j)d}} \right) \\
    &= O \left( \sqrt{n} 2^{d/2} \left( \log(n) \sum_{j=1}^h 2^{(\frac{d}{2}-1)(h-j)} - d \sum_{j=1}^h (h-j) 2^{(\frac{d}{2}-1)(h-j)} \right) \right) \\
    &= O \left( \sqrt{n} 2^{d/2} \left( \log(n) \sum_{j=0}^{h-1} 2^{(\frac{d}{2}-1)j} - d \sum_{j=0}^{h-1} j 2^{(\frac{d}{2}-1)j} \right) \right)
\end{align*}
for the $2^d$-ary $2$-HST with height $h$.
When $d = 1$,
\begin{align*}
    2^{d/2} \left( \log(n) \sum_{j=0}^{h-1} 2^{(\frac{d}{2}-1)j} - d \sum_{j=0}^{h-1} j 2^{(\frac{d}{2}-1)j} \right)
    = O \left( \log(n) \sum_{j=0}^{h-1} 2^{-j/2} - \sum_{j=0}^{h-1} j 2^{-j/2} \right)
    = O(\log n).
\end{align*}
When $d = 2$,
\begin{align*}
    2^{d/2} \left( \log(n) \sum_{j=0}^{h-1} 2^{(\frac{d}{2}-1)j} - d \sum_{j=0}^{h-1} j 2^{(\frac{d}{2}-1)j} \right)
    = O \left( h\log(n) - h(h-1) \right)
    = O(\log^2 n).
\end{align*}
When $d \geq 3$,
\begin{align*}
    &2^{d/2} \left( \log(n) \sum_{j=0}^{h-1} 2^{(\frac{d}{2}-1)j} - d \sum_{j=0}^{h-1} j 2^{(\frac{d}{2}-1)j} \right) \\
    &= 2^{d/2} \left( \log(n) \cdot \frac{2^{(\frac{d}{2}-1)h} - 1}{2^{\frac{d}{2}-1} - 1} - d \left( \frac{(h-1)2^{(\frac{d}{2}-1)h}}{2^{\frac{d}{2}-1}-1} - \frac{2^{(\frac{d}{2}-1)h} - 2^{\frac{d}{2}-1}}{(2^{\frac{d}{2}-1} - 1)^2} \right) \right) \\
    &\leq 2^{d/2} \left( (\log(n) - dh) \cdot \frac{2^{(\frac{d}{2}-1)h}}{2^{\frac{d}{2}-1} - 1} + \frac{d2^{(\frac{d}{2}-1)h}}{2^{\frac{d}{2}-1}-1} + \frac{d2^{(\frac{d}{2}-1)h}}{(2^{\frac{d}{2}-1} - 1)^2} \right) \\
    &= O(d2^{(\frac{d}{2}-1)h})
    = O(dn^{\frac{1}{2} - \frac{1}{d}}),
\end{align*}
where the second equality holds since $h \geq \log(n) / d - 1$, concluding the proof.

\end{document}